\documentclass[conference]{IEEEtran}
\usepackage{color,graphicx,amsmath,amssymb,amsthm,epsfig,subfigure,mathrsfs,cite}
\begin{document}
\title{Deterministic Constructions for Large Girth Protograph LDPC Codes}

\author{
 \IEEEauthorblockN{Asit Kumar Pradhan}
 \IEEEauthorblockA{Dept. of Electrical Engg.\\
   IIT Madras, Chennai, India\\
   Email: ee11m004@ee.iitm.ac.in
   }
 \and
 \IEEEauthorblockN{Arunkumar Subramanian}
 \IEEEauthorblockA{Link-A-Media Devices Corp.\\
   Santa Clara, CA, USA
   } 
 \and
 \IEEEauthorblockN{Andrew Thangaraj}
 \IEEEauthorblockA{Dept. of Electrical Engg.\\ 
   IIT Madras, Chennai,India\\
   Email: andrew@ee.iitm.ac.in}}

\newtheorem{theorem}{Theorem}

\maketitle

\begin{abstract}
	For certain degree-distribution pairs with non-zero fraction of degree-two bit nodes, 
the bit-error threshold of the standard ensemble of Low Density Parity
Check (LDPC) codes is known to be close to capacity. However, the degree-two bit
nodes preclude the possibility of a block-error
threshold. Interestingly, LDPC codes
constructed using protographs allow the possibility of having both degree-two bit nodes and a block-error
threshold. In this paper, we analyze density evolution for protograph
LDPC codes over the binary erasure channel and show that their bit-error probability decreases double
exponentially with the number of iterations when the erasure probability is below the bit-error threshold and long chain of degree-two variable nodes are
avoided in the protograph. We present deterministic constructions of such
protograph LDPC codes with girth logarithmic in
blocklength, resulting in an exponential fall in bit-error
probability below the threshold. We provide optimized protographs, whose block-error
thresholds are better than that of the standard ensemble  with minimum
bit-node degree three. These protograph LDPC codes are theoretically
of great interest, and have applications, for instance, 
in coding with strong secrecy over wiretap channels.
\end{abstract} 

\section{Introduction}

Constructing a sequence of codes with efficient encoders/decoders and a guarantee that block-error
rate tends to zero with increasing block-length is one of the major
goals of coding theory. At rates below capacity, such ``good'' sequences are
known to exist, but many classical code sequences do not have this
property. Modern code constructions, such as Low Density Parity Check
(LDPC) codes, define a sequence of ensembles of codes with efficient
decoders and probabilistic concentration results that come close to achieving the
goal of constructing good code sequences \cite{mct}. Recently, polarization \cite{5075875} and spatial coupling \cite{5695130} have
been used to construct good code sequences for binary symmetric channels.

In this work, we are primarily interested in deterministic
constructions of sequences of
good LDPC codes with block-error thresholds nearing capacity limits. We will stick
to the binary erasure channel, though the work can be extended to
other binary-input symmetric channels. Most of the prior work in this
area provides probabilistic guarantees on ensembles of LDPC codes, and
most of these guarantees are for bit-error probabilities. The block-error threshold problem for LDPC codes was first studied in
\cite{Lentermaier}, where standard ensembles with a minimum bit-node degree (denoted
$l_{\min}$) of three was shown to have
block-error thresholds. For the standard irregular ensemble with $l_{\min}=2$, the
block-error rate, surprisingly, tends to a constant as block-length
increases. The main cause for this problem is the presence of long chains of
degree-two nodes in the standard ensemble. However, (bit-error)
capacity-approaching LDPC degree distributions
have a significant fraction of degree-two bit nodes. For instance, the
best threshold for rate-1/2 codes with minimum left degree three is
only about 0.461 leaving a significant gap to the capacity threshold of
0.5. So, while degree-two nodes are needed to approach capacity, they
preclude the possibility of a block-error threshold. One of the goals of this work is to construct LDPC codes with
block-error thresholds that improve this gap to capacity.

A key idea in the construction of LDPC code ensembles with degree-two nodes
and decaying block-error performance is the notion of multi-edge type
(MET) ensembles \cite{Richardson,mct}, of which the protograph
LDPC code ensemble \cite{Thrope} has received considerable practical attention
because of ease of implementation. In \cite{Richardson}, the standard ensemble is
restricted in a suitable fashion to limit the impact of degree-two
nodes. In \cite{Thrope}, density evolution and optimization
for protograph LDPC code ensembles was described and carried out. In
\cite{5174517}, protographs are optimized for thresholds nearing
capacity, and linear growth of ensemble-averaged weight distribution is
established for protograph LDPC code
ensembles. There have been numerous other work in the construction of
protographs in practical implementations.

The use of large-girth graphs in constructing LDPC codes started with
Gallager's thesis \cite{gallagerthesis}, where regular LDPC codes with
large girth were constructed. The Lubotzky-Phillips-Sarnak (LPS)
constrction \cite{LPS} of Ramanujan graphs has been used in the
construction of regular and irregular LDPC codes in
\cite{Rosenthal}. As shown in \cite{arunForensics}, large-girth LDPC codes with minimum left
degree, $l_{\min}>2$, achieve an exponential decay of bit error, i.e
$\mathcal{O}(\text{exp}(-c_1n^{c_2}))$ for constants $c_1$, $c_2$, 
over a binary erasure channel BEC$(\epsilon)$, when $\epsilon$
is less than the density evolution threshold $\epsilon^*$. So,
large-girth LDPC codes have a block-error threshold equal to their
bit-error thresholds, when $l_{\min}>2$.

In this work, we provide deterministic constructions for a sequence of
good LDPC codes by using large-girth graphs along with suitable
protographs that contain degree-two nodes. This allows us to achieve
BEC block-error thresholds as high as 0.486 with small ($8\times 16$)
protographs. To do this, we begin by studying the
density evolution for protograph ensembles, and show that
bit-error decreases double exponentially in number of iterations at
erasure probabilities smaller than the threshold even
when $l_{\min}=2$, if long chains of degree-two variable nodes are
avoided through the protograph. To avoid chains of degree-two nodes,
we allow at most one degree-two variable node to connect to a check 
node in the protograph. 

We then provide a construction for large-girth protograph LDPC codes
starting with a given large-girth regular graph and performing
suitable node splitting operations. Using the LPS Ramanujan
graphs, we provide deterministic constructions of
large-girth protograph LDPC codes that achieve
an exponential decay for bit-error probability with blocklength and
still contain degree-two bit nodes. In comparison with prior work, we have analyzed the
density evolution for protograph LDPC codes directly and showed the
double-exponential decay with iterations even in the presence of
degree-two nodes. Our node splitting construction is more general than
the one in \cite{Rosenthal} and provides a deterministic construction
with guaranteed block-error probability behavior.

\section{Protograph LDPC Codes}
\label{sec:protograph-ldpc-code}
Following the notation in \cite{Thrope}, a protograph $G=(V,C,E)$ is a
bipartite graph with $V$ and $C$ being the sets of variable and check nodes, respectively, and $E$ being the set of undirected edges that connect a vertex in $V$ to a vertex in $C$. Multiple edges are allowed between a pair of nodes $(v,c)\in V\times C$. 

A protograph can be represented by a base
matrix $B$, where $B(i,j)$ is the number of edges between the $i$-th
check node (denoted $c_i$) and the $j$-th variable node (denoted $v_j$). For example, consider a base matrix 
\begin{align}
 \nonumber B=
 \begin{bmatrix}
  1 & 1 & 1 & 0 \\
  1 & 1 & 1 & 1\\
  0 & 1 & 1 & 1
 \end{bmatrix}.
\end{align}
The protograph corresponding to the above base matrix is shown in Fig. \ref{fig:a}.
\begin{figure}[htb]
\centering
\subfigure[Example of a protograph.]
{
 \includegraphics[scale=0.55]{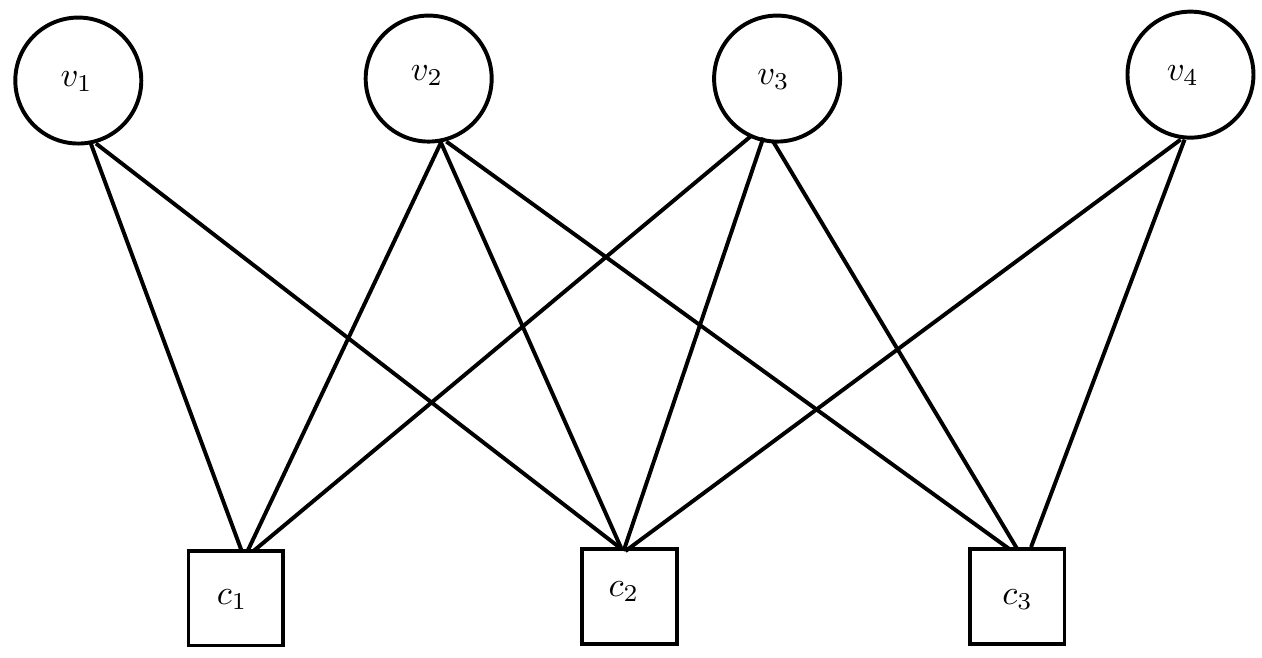}
\label{fig:a}
}
\subfigure[Computation graph $C_1(v_1)$.]
{
 \includegraphics[scale=0.4]{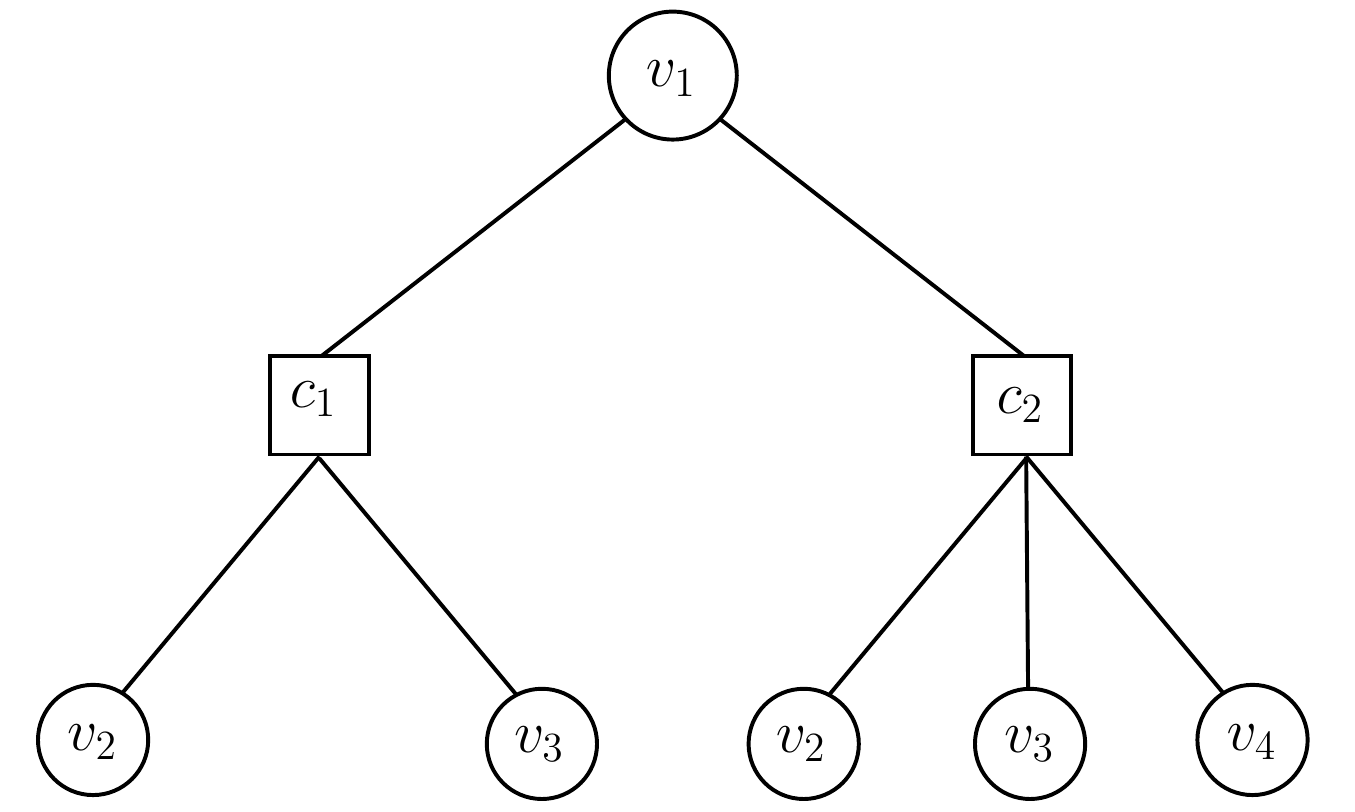}
\label{fig:b}  
}
\caption{Protograph and computation graph.}
\label{fig:protex}
\end{figure}

\subsection{Lifted Graphs}
We can apply a copy and permute operation to a protograph to obtain expanded or lifted graphs of different sizes \cite{Thrope}. A given protograph $G$ is copied, say $T$ times, with the $t$-th copy having nodes $(v,t)$ and $(c,t)$, and edges $(e,t)$. Then, for each edge $e$ in the protograph, we assign a permutation $\pi_e$ of the set $\{1,2,\cdots,T\}$. In the permute operation, an edge $(e,t)$ connecting $(v,t)$ and $(c,t)$ is permuted so as to connect variable node $(v,t)$ to check node $(c,\pi_e(t))$. We will denote the lifted graph as $G'=(V',C',E')$. The lifted graph of a protograph can be thought of as a Tanner graph of an LDPC code, which is referred to as a protograph LDPC code. In general, $( T!)^{|E|}$ lifted graphs or protograph LDPC codes can be obtained from a protograph, each corresponding to a different permutation, where $|V|$ is the number of variable nodes in the protograph. The collection of these lifted graphs is called the protograph ensemble of LDPC codes. Protograph LDPC codes are a special class of MET-LDPC codes, with each edge in the protograph being of a different type. The degree distribution of check and variable nodes in the lifted graph is the same as that of the protograph. So, the (designed) rate of the protograph LDPC code is given by $1-\frac{|C|}{|V|}$, where $|C|$ and $|V|$ denote the number of check and variable nodes in the protograph, respectively. 

Let $v$ be a variable node in the lifted graph $G'$. The $t$-iteration computation graph associated with $v$, denoted $C_t(v)$, is defined as the subgraph of $G'$ obtained by traversing from $v$ up to the $t$-th iteration level along all edges \cite{mct}. The structure of the computation graph is completely determined by the protograph $G$ for all lifted graphs $G^{'}$. 
An example of a computation graph is shown in Fig. \ref{fig:b}. 

\subsection{Density Evolution for Protograph Codes}
Let us consider the standard message-passing decoder over a binary
erasure channel (BEC$(\epsilon)$) run on a lifted graph $G'$ derived
from a protograph $G=(V,C,E)$. Since the lifted graphs form an MET
ensemble with $|E|$ edge types, density evolution proceeds with $|E|$
erasure probabilities, one for each edge in the protograph
\cite{Richardson2004}. Let $E=\{e_1,e_2,\ldots,e_{|E|}\}$ with edge
$e\in E$ connecting variable node $v_e$ with check node $c_e$. Let
$x_t(i)$ be the probability that an erasure is sent from variable node
to check node along edge $e_i$ in the $t$-th iteration. Similarly, let
$y_t(j)$ be the probability that an erasure is sent from check node to
variable node along edge $e_j$ in the $t$-th iteration. The density evolution recursion \cite{mct} is given by
$$x_0(i) =\epsilon,$$
\begin{equation}
\label{eq:1}y_{t+1}(j)=1-\prod_{i\in E_c(e_j)}(1-x_{t}(i)), \;\forall t\geq1,
\end{equation}
\begin{equation}
\label{eq:2}x_{t+1}(i)=\epsilon\prod_{j\in E_v(e_i)}y_{t+1}(j), \;\forall t\geq1,
\end{equation}
where $E_c(e_j)=\{i \neq j:c_{e_j}=c_{e_i}\}$ is the set of all indices of edges adjacent to the same check node as the edge $e_j$, and $E_v(e_i)=\{j \neq i:v_{e_i}=v_{e_j}\}$ is the set of all indices of edges adjacent to the same variable node as the edge $e_i$. The density evolution threshold, denoted $\epsilon_{\text{th}}$, for the protograph-based LDPC code ensemble is defined as the largest value of $\epsilon$ for which erasure probability on each edge of the protograph tends to zero, as $t\rightarrow \infty$. i.e. $\epsilon_{\text{th}}=\sup\{\epsilon:\max_{i}x_t(i)\to0\}$. Clearly, this is also the threshold value below which the erasure probability of a variable node in the protograph goes to zero with increasing iterations (as long as $l_{\min} > 1$).

\subsection{Asymptotic Behavior of Density Evolution}
\label{sec:asympt-behav-dens}
\begin{theorem}For $\epsilon<\epsilon_{\text{th}}$, $\max_i x_t(i)$
exhibits a double-exponential decay with $t$ when
not more than one degree-two variable node is connected to a check node in the protograph.
\end{theorem}
\begin{proof}
We will repeatedly use the following inequality. For any $x\in[0,1]$
and a positive integer $d$, 
\begin{align}
\label{eq:3}(d-1)x\geq1-(1-x)^{d-1}.
\end{align}
Let $d$ be the maximum degree of a check node in the protograph, and
let $\bar{x}_t=\max_i x_t(i)$. Since $\epsilon<\epsilon_{\text{th}}$, we
have $\bar{x}_t\to0$. We pick $t$ large enough to have
$0\leq(d-1)\bar{x}_t<1$. 

We note that 
\begin{align*}
	1 - x_t(j) &\geq 1 - \bar{x}_t, \quad \forall j \\
	\Rightarrow \prod_{j \in E_c(e_i)} \left( 1 - x_t(j) \right) &\geq \left( 1 - \bar{x}_t \right)^{(d-1)}, \quad \forall i
\end{align*}

Using (\ref{eq:1}) and (\ref{eq:3}), we get 
\begin{align}
	y_{t+1}(j) &\leq (d-1) \bar{x}_t, \quad \forall j \label{eq:4} \\
	\Rightarrow x_{t+1}(i) &= \epsilon \prod_{e_j \in E_v(e_i)} y_{t+1}(j)  \notag\\
	&\leq \epsilon (d-1)^{(l_m - 1)} \left( \bar{x}_t \right)^{(l_m - 1)}, \quad \forall i \label{eq:5}
\end{align}

Repeating the process, we get 
\begin{align*}
	y_{t+2}(j') &\leq (d-1) x_{t+1}(i) \leq \epsilon (d-1)^{l_m} \left( \bar{x}_t \right)^{(l_m - 1)}\\
	x_{t+2}(i') &\leq \epsilon^{l_m'} (d-1)^{l_m(l_m' - 1)} \left( \bar{x}_t \right)^{(l_m - 1) (l_m' - 1)} \\
	\Rightarrow \bar{x}_{t+2}&\leq \epsilon^{l_m'} (d-1)^{l_m(l_m' - 1)} \left( \bar{x}_t \right)^{(l_m - 1) (l_m' - 1)} 
\end{align*}
where $l_m$ and $l_m'$ are the minimum degrees of variable nodes in $E_v(e_{i})$ and $E_v(e_{i'})$, respectively.
Since a check node is connected to at most one degree-two
variable node, we have that either $l_m\ge3$ or $l_m'\ge3$. So,
$(l_m-1)(l_m'-1)\ge2$. So, we have
\begin{align}
	\label{eq:6}\bar{x}_{t+2}\leq A (\bar{x}_t)^2,
\end{align}
where $A=\epsilon^{l_m'}(d-1)^{l_m(l_m' - 1)}$ does not vary with $t$. Since $\epsilon < \epsilon_{\text{th}}$ and $\bar{x}_t\to0$, there 
exists an $R$ such that $A(\bar{x}_{R})^2< 1$ and $ (d-1)\bar{x}_R<
1$. Following arguments similar to those in \cite{arunForensics}, we
can show that 
\begin{equation}
  \label{eq:9}
\bar{x}_{R+2i}\leq (A\bar{x}_R)^{2^i},  
\end{equation}
which implies a double-exponential decay of $\bar{x}_t$ with $t$ (for
sufficiently large $t$).
\end{proof}

For the standard ensemble of LDPC codes, density evolution analysis is
approximate because of the following assumptions:
\begin{itemize}
 \item[1.] The computation graph is a tree.
 \item[2.] The node degrees in the computation graph are independent.
\end{itemize}
Typically, probabilistic concentration results and asymptotic
guarantees are used to support the practical validity of density
evolution in the standard ensemble. In this work, we use protograph
codes and MET density evolution --- these make Assumption 2 above
unnecessary. For Assumption 1 to be true, we consider large-girth graphs. 
\section{Large-Girth Protograph LDPC Codes}
\label{sec:large-girth-prot}
To achieve strong secrecy, Tanner graphs whose girth increases as
$\log n$ with blocklength $n$ have been used in
\cite{arunForensics}. In this work, we extend the construction in
\cite{arunForensics} to construct bipartite graphs from protographs
with girth increasing as $\log n$. Once girth is $\Theta(\log n)$, message
error rates for iterations
up to $\Theta(\log n)$ will exactly follow the protograph density evolution
of (\ref{eq:1})-~(\ref{eq:2}). Following the analysis in Section
\ref{sec:asympt-behav-dens}, the message error rate falls
double-exponentially in $\log n$, or exponentially in block-length
$n$. This results in an inverse polynomial decay, $\mathcal{O}(1/n^k)$ for any
$k$, for block-error rate.

\subsection{ LPS Graphs $X^{p,q}$} 
While the construction method can
use any sequence of regular large-girth graphs, one explicit
possibility is the LPS
construction \cite{LPS}. LPS graphs belong to the class
of Cayley graphs. Given a group $G$ and an inverse-closed subset $S$ of
$G$, i.e, $ s^{-1}\in S$,$\forall s \in S$, the Cayley graph
($\Gamma(G,S)$) is the undirected simple graph defined as follows:
\begin{itemize}
\item The vertex set of $\Gamma(G,S)$ is $G$.
\item For any $g\in G$ and $s \in S$, there is an edge between $g$ and $gs$.
\end{itemize}

Let $p$ and $q$ be distinct, odd primes with $q>2\sqrt{p}$. The LPS
graph, denoted $X^{p,q}$ \cite{LPS}, is a connected, $(p+1)$-regular graph and has
the following properties:
\begin{itemize}
 \item If $p$ is a quadratic residue $\mod q$, then $X^{p,q}$ is a
   non-bipartite graph with $q(q^2-1)/2$ vertices and girth
   $g(X^{p,q}) \geq 2\log_p q$.
\item If $p$ is a quadratic non-residue $\mod q$, then $X^{p,q}$ is a
  bipartite graph with $q(q^2-1)$ vertices and girth
  $g(X^{p,q}) \geq 4\log_p q-\log_p 4$.
\end{itemize}
When $X^{p,q}$ is non-bipartite, we can convert it to a bipartite
graph using the following algorithm \cite{arunForensics}:
\begin{itemize}
 \item Given a graph $G$ with vertices $V(G)$ and edges $E(G)$, construct a copy
   $G^\prime$ with a new vertex set $V(G')$ and a new edge set $E(G')$. Let $f:V(G)\rightarrow
   V(G^\prime)$ be the 1-1 mapping from a vertex in $G$ to its copy in $G'$.
\item Create a bipartite graph $H$ with vertex set $V(G)\cup V(G^\prime)$ and
  edge set $E(H)=\{(x,f(y)):(x,y)\in E(G)\}$.
\end{itemize}
Following \cite{B.Bollobas}, it was shown in \cite{arunForensics} that $g(H)\geq g(G)$.

For constructing a sequence of $d$-regular large-girth graphs for an arbitrary $d$ using the LPS graphs, we use the following trick from \cite{arunForensics}. There exists an infinite number of primes $p$ such that $d|(p+1)$. For each such prime $p$ and a suitable $q$, we construct $X^{p,q}$ and split each $(p+1)$-degree node into $(p+1)/d$ nodes of degree $d$. As shown in \cite{arunForensics}, the node splitting does not reduce girth and we have a large-girth graph of the required degree $d$.
 \subsection{Node Splitting for MET-LDPC Codes}
The construction of a large-girth protograph LDPC code starts with a $d$-regular large-girth bipartite graph
$G$ with $d$ being the number of edges in the protograph. The
bipartition of $G$ will contain $|V(G)|/2$ left and right vertices. 
We associate $d$ sockets with each vertex of $G$, and associate each edge
connected to a vertex with one of the sockets. 

According to K\"onig's theorem, the edge chromatic number of a bipartite graph is equal to the maximum
degree of its nodes. Therefore $G$ has an edge coloring involving $d$
colors. Based on this edge coloring, we define a coloring of the
sockets in $G$ by the colors $S = \left\{ s_1, s_2, \ldots, s_d
\right\}$. Let $P$ and $Q$ be two fixed partitions of $S$, with
$P=\{P_1,P_2,\ldots, P_l\}$ and $Q=\{Q_1,Q_2,\ldots,Q_r\}$. If $G$ is a
Cayley graph, then the colors $S$ can be associated with the generating set,
  along with a direction. 

The main step in the construction is splitting the left and right vertices of $G$ according to $P$ and $Q$,
respectively. A left vertex $v$ is split into sub-vertices $v_1, v_2,
\ldots, v_l$, such that for any $i$, the sockets of $v$ in $P_i$ get
associated with $v_i$. A right vertex $c$ is split into sub-vertices
$c_1, c_2, \ldots, c_r$, such that for any $j$, the sockets of $c$ in
$Q_j$ get associated with $c_j$. The resulting Tanner graph, denoted
$T(G,P,Q)$, will have $l|V(G)|/2$ variable nodes and $r|V(G)|/2$ check
nodes, and the associated MET-LDPC code will have design rate $1-\frac{r}{l}$.

We note the following important properties of $T(G, P, Q)$.
\begin{enumerate}
\item It was shown in \cite{arunForensics} that the above node-splitting
procedure does not decrease girth. So, the girth of $T(G,P,Q)$ is not less than 
the girth of $G$ for any $P$ and $Q$.

\item The Tanner graph $T(G,P,Q)$ is, in fact, a lifted version of
a protograph with $l$ variable nodes indexed by $P_i$, $1\le i\le l,$ and $r$ check
nodes indexed by $Q_j$, $1\le j\le r$. Variable node $P_i$ in the
protograph is connected by an edge to a check node $Q_j$, whenever
$P_i\cap Q_j\ne \emptyset$. So, the number of edges in the protograph
is $|S|$.

\item The protograph is copied $|G|/2$ times and the edge permutation
  is induced by the edge connections of the original graph $G$. 
\end{enumerate}

The procedure to generate a sequence of large-girth protograph LDPC
codes can be summarized as follows. By fixing the degree $d$ and the 
partitions $P$ and $Q$, we fix a protograph. We then apply the above 
node-splitting procedure to a sequence of large-girth $d$-regular
bipartite graphs. This results in a sequence of large-girth Tanner graphs that 
are liftings of the protograph defined by $(d,P,Q)$. From the above construction, we have the following theorem. 
\begin{theorem}
For a given protograph with threshold $\epsilon_{\text{th}}$, there
exists a deterministic sequence of large-girth liftings with
increasing length $n$ such that block error probability falls as $ne^{-cn}$, for a constant $c>0$, over a BEC$(\epsilon)$ with $\epsilon<\epsilon_{\text{th}}$.
\end{theorem}

\subsection{An Example} \label{sec:proto-example}
Let $d = 12$, and let the socket colors 
$  S=\{s_1,s_2,\cdots,s_{12}\} $
be split into two partitions $P$ and $Q$ given by 
\begin{align*}
P_1 &= \{ s_1,s_2\} & P_2 &= \{ s_3,s_4,s_5\}\\
P_3 &= \{ s_6,s_7,s_8\} & P_4 &= \{ s_9,s_{10},s_{11},s_{12}\}
\end{align*}
\begin{align*}
	Q_1 &= \left\{ s_1, s_3, s_6, s_9, s_{10}, s_{11}\right\} \\
	Q_2 &= \left\{ s_2, s_4, s_5, s_7, s_8, s_{12} \right\}
\end{align*}

The variable nodes are split as shown in Fig.~\ref{a}. The protograph generated by this choice of $(d, P, Q)$ is shown in Fig.~\ref{b}. The design rate of this protograph is $1/2$. 
\begin{figure}
\centering
\subfigure[]
{
 \includegraphics[scale=0.25]{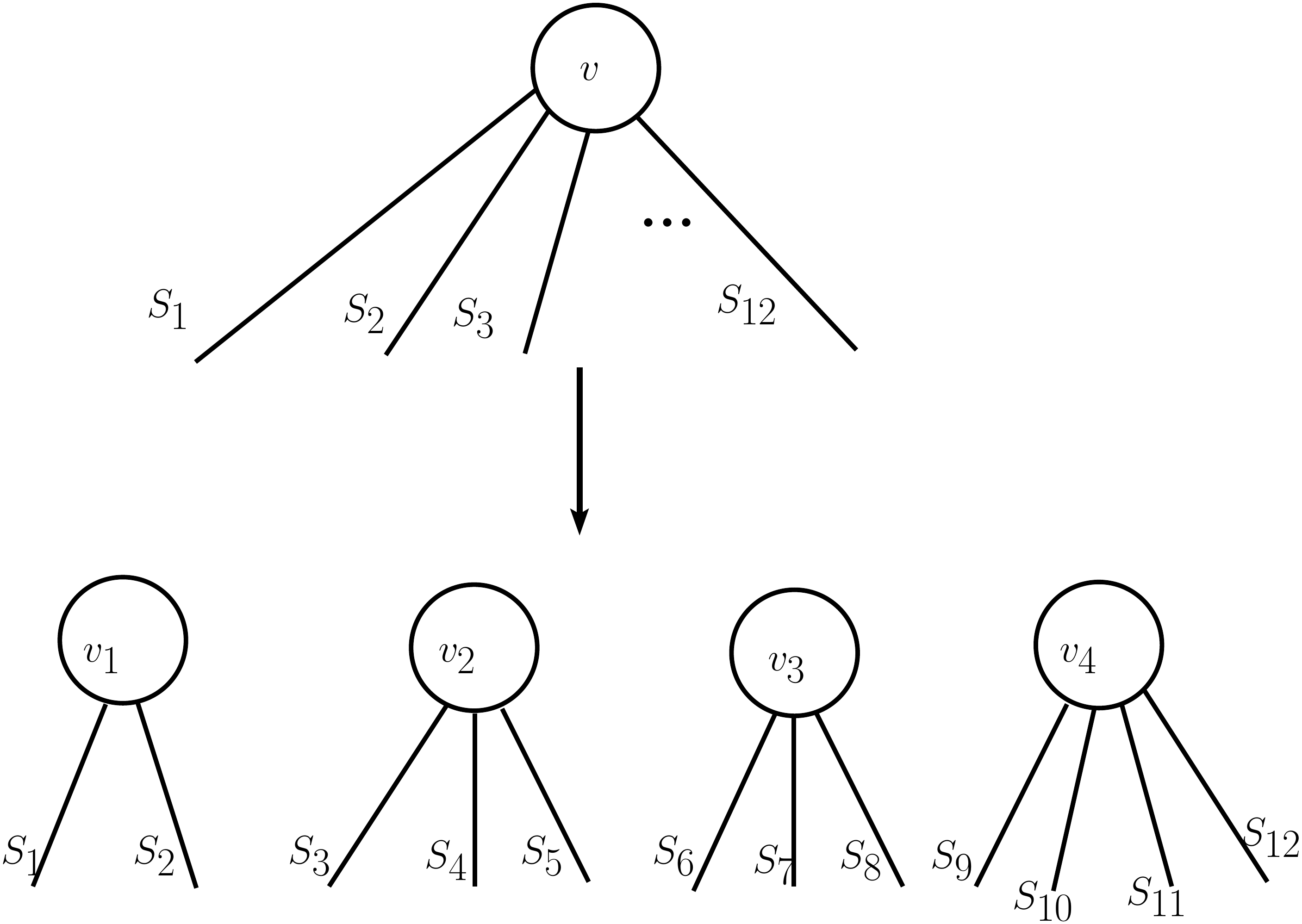}
\label{a}
}
\subfigure[]
{
 \includegraphics[scale=0.24]{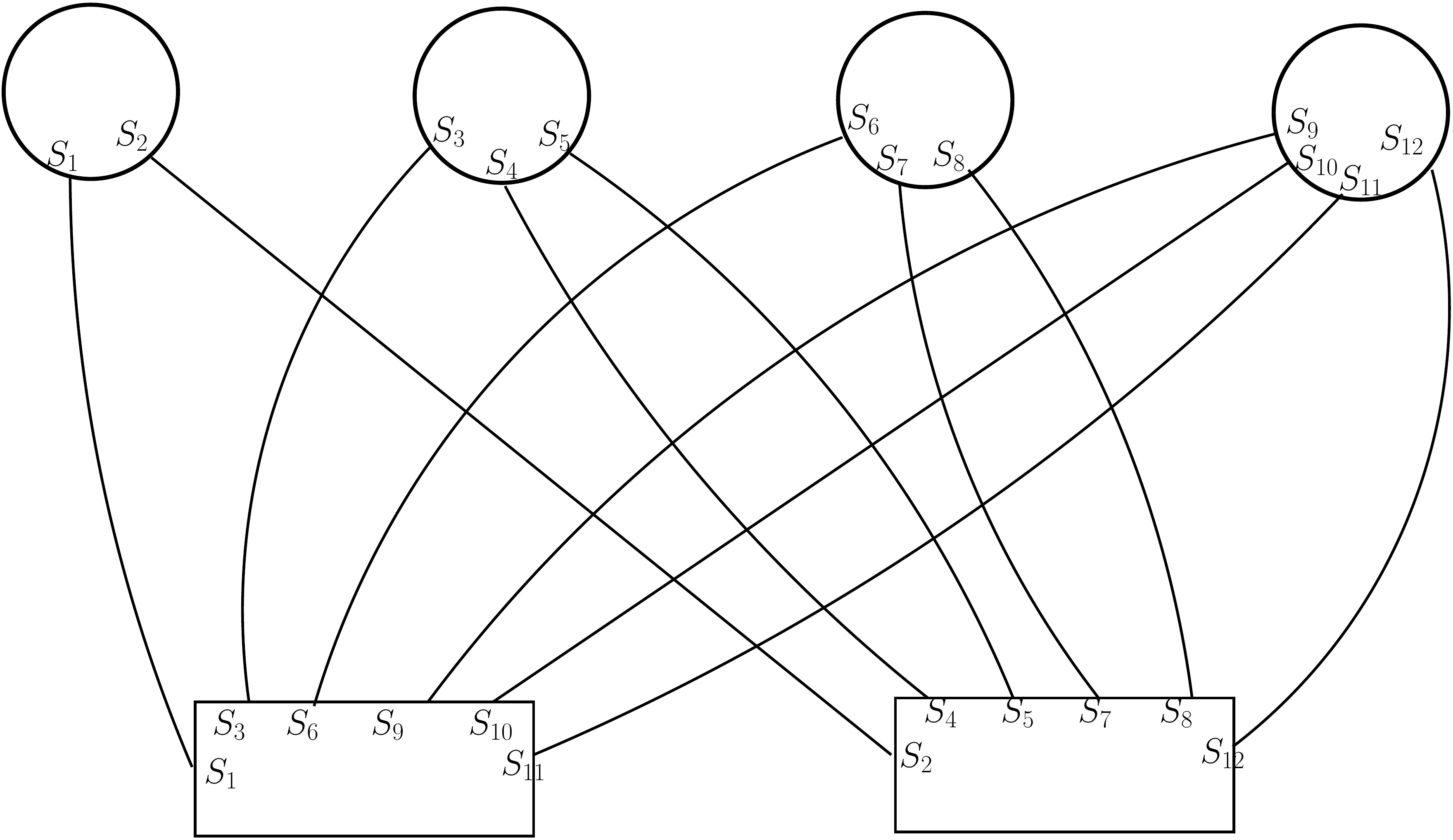}
\label{b}  
}
\caption{Illustration of (a) variable node splitting, and (b) protograph for \ref{sec:proto-example}.}
\end{figure}
\section{Optimization of Protographs}
\label{sec:optim-prot}
We have optimized protographs using differential evolution \cite{DE}\cite{Storn}, where
we use the threshold given by density evolution as the cost function. The
salient steps of the differential evolution algorithm are described briefly in the following:
\begin{itemize}
 \item[1.] Initialization: For generation $G=0$, we randomly choose
   $N_P$ base matrices $B_{k,G}$, with $0\le k\le N_P-1$, of size $|C|\times|V|$, where
   $N_P=10|C||V|$. Each entry of $B_{k,G}$ is binary, chosen
   independently and uniformly. 
\item[2.] Mutation: Protographs of a particular generation are interpolated as follows.
\begin{align}
 M_{k,G}=[B_{r_1,G}+0.5(B_{r_2,G}-B_{r_3,G})],
\end{align}
where $r_1$, $r_2$, $r_3$ are randomly-chosen distinct values in the
range $[0,N_P-1]$, and $[x]$ denotes the absolute value of $x$ rounded to
the nearest integer.
\item[3.] Crossover: A candidate protograph $B'_{k,G}$ is chosen as
  follows. The $(i,j)$-th entry of $B'_{k,G}$ is set as the $(i,j)$-th
  entry of $M_{k,G}$ with probability $p_c$, or as the $(i,j)$-th
  entry of $B_{k,G}$ with probability $1-p_c$. We use $p_c=0.88$ in
  our optimization runs. In $B'_{k,G}$, if any check node is connected to
  more than one degree-two variable node, edges are reassigned. So,
  each $B'_{k,G}$ avoids long chain of degree-two variable nodes.
\item[4.] Selection: For generation $G+1$, protographs are selected as
  follows. If the threshold of $B_{k,G}$ is greater than that of
  $B'_{k,G}$, set $B_{k,G+1}=B_{k,G}$; else, set $B_{k,G+1}=B'_{k,G}$.
\item[5.] Termination: Steps 2--4 are run for several generations (we
  run up to $G=6000$) and the protograph that gives the best threshold is chosen as the optimized protograph.
\end{itemize}
Results from our optimization runs are given in Table \ref{tab:optprot}.
We see that the optimized protographs give better thresholds than
irregular standard ensemble codes with minimum degree 3.
 An optimized $4\times 8$ protograph with threshold 0.479 is given by the following base matrix:
\begin{equation}
  \label{eq:8}
\begin{bmatrix}
1 & 2 & 2 & 3 & 4 & 1 & 1 & 0  \\
0 & 1 & 0 & 0 & 5 & 0 & 0 & 1 \\
1 & 0 & 0 & 0 & 3 & 0 & 4 & 1\\
1 & 0 & 1 & 0 & 6 & 1 & 0 & 0 
\end{bmatrix}
\end{equation}
An optimized $8\times 16$ protograph with threshold 0.486 is given by the following base matrix:
\begin{equation}
  \label{eq:7}
\begin{bmatrix}
1 & 2 & 0 & 0 & 1 & 0 & 0 & 4 & 0 & 0 & 0 & 0 & 0 & 0 & 0 & 1\\
0 & 1 & 0 & 0 & 0 & 1 & 0 & 0 & 2 & 2 & 1 & 0 & 0 & 0 & 1 & 1\\
0 & 3 & 1 & 2 & 1 & 0 & 0 & 0 & 4 & 0 & 0 & 3 & 2 & 2 & 0 & 3\\
0 & 5 & 0 & 0 & 0 & 0 & 1 & 1 & 0 & 0 & 1 & 0 & 0 & 1 & 0 & 0\\
1 & 3 & 1 & 1 & 1 & 2 & 0 & 0 & 1 & 0 & 0 & 0 & 0 & 0 & 0 & 0\\
1 & 5 & 0 & 0 & 0 & 3 & 1 & 0 & 0 & 0 & 1 & 0 & 0 & 0 & 0 & 0\\
0 & 4 & 0 & 0 & 0 & 0 & 0 & 1 & 1 & 0 & 0 & 0 & 0 & 0 & 0 & 1\\
0 & 5 & 0 & 0 & 0 & 0 & 0 & 0 & 0 & 1 & 0 & 0 & 1 & 0 & 1 & 0
\end{bmatrix}  
\end{equation}
Note that the above protographs have block-error threshold same as the
bit-error threshold, and the block-error rate falls inverse polynomially in
block-length under the large girth construction as described in
Sections \ref{sec:protograph-ldpc-code} and
\ref{sec:large-girth-prot}. 

\begin{table}[tbh]
  \centering
  \begin{tabular}{|c|c|}
    \hline
    {\bf Code type}&{\bf Threshold}\\
    \hline
    Standard ensemble ($l_{\min}=3$)&0.461\\
    \hline
    $4\times 8$ protograph in (\ref{eq:8})&0.479\\
    \hline
    $8\times 16$ protograph in (\ref{eq:7})&0.486\\
    \hline
  \end{tabular}
  \caption{Optimized protographs and thresholds (rate 1/2).}
  \label{tab:optprot}
\end{table}

In future work, we hope to obtain better thresholds that are closer to
capacity limits,
by further increasing the size of the
protograph.
\section{Conclusion}
\label{sec:conclusion}
In this work, we presented a deterministic construction for a sequence
of codes for the binary erasure channel with block-error rate
falling inverse polynomially with block-length at rates close to the
capacity. The codes are protograph LDPC codes that avoid long degree-two
variable node chains and are constructed from large-girth graphs. To the best
of our knowledge, this is the first deterministic construction of LDPC.
codes with guaranteed block-error thresholds nearing capacity limits.
   
\bibliographystyle{IEEEtran}
\bibliography{IEEEabrv,references}

\end{document}